\documentclass{aptpub}
\usepackage{subfig}
\usepackage{graphicx}
\usepackage{float}
\usepackage{amscd}
\usepackage{amsmath}
\usepackage{amssymb}

\authornames{John Lombard} 
\shorttitle{Novel Algorithms for Sampling Abstract Simplicial Complexes} 

\begin{document}
\title{Novel Algorithms for Sampling Abstract Simplicial Complexes}

\authorone[University of Washington]{John Lombard}
\addressone{Department of Physics, University of Washington, Seattle, Washington, 98195, USA. Email Address: jml448@uw.edu}


\begin{abstract}

  We provide dual algorithms for sampling the space of abstract simplicial complexes on a fixed number of vertices. We develop a generative and descriptive sampler designed with heuristics to help balance the combinatorial multiplicities of the states and more widely sample across the space of nonisomorphic complexes. We provide a formula for the exact probabilities with which this algorithm will produce a requested labeled state, and compare with an existing benchmark. We also design a highly conductive local ergodic random walk with known transition probabilities. We characterize the autocorrelation of the walk, and numerically test it against our sampler to illustrate its efficacy.
\end{abstract}


\section{Introduction to the Space and Use of Abstract Simplicial Complexes}

  Whether used to model information theoretic phenomena like social networks or to study the combinatorial properties of fundamental structures in understanding emergent geometry, abstract simplicial complexes have a rich history of applications and are increasingly used in physics as powerful tools with extensive mathematical structures \cite{Maletic2009}. Unlike 1-dimensional graphs that only convey connectivity data between nodes, abstract simplicial complexes (ASCs) are generalizations that can allow representations of data through higher-dimensional geometric structures, such as surfaces and volumes in the form of combinatorial triangles and tetrahedra (and their higher dimensional equivalents). Informally, an ASC is the combinatorial abstraction of a geometric simplicial complex encoding the downward closure property. Unlike a geometric simplicial complex where the intersections of any two simplicies in the complex must also be a simplex in the complex that is in the union of the boundaries of the intersecting simplicies, ASCs only require that any boundary of a simplex is also a simplex in the complex. For example, the clique complex of a graph---the set of all complete subgraphs---is an abstract simplicial complex on the vertices. For a graphical picture of the differences of an ASC with a geometric simplicial complex when embedded into an ambient space, see Fig. \ref{fig:ASC}.

  \begin{figure}[H]
  \centering
  \subfloat{(a)\includegraphics[scale=4]{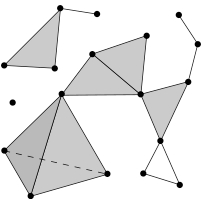}}(b)\subfloat{\includegraphics[scale=3]{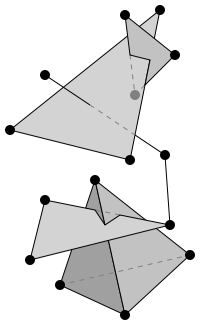}}
  \caption[A Simplicial Complex and a Clique Complex, both Embedded in $\mathbb{R}^{3}$]{A Simplicial Complex (a) and a Clique Complex (b), both Embedded in $\mathbb{R}^{3}$}\label{fig:ASC}
  \centering
  \end{figure}

  This structure allows one to model more complex association data that may not be captured by the limited degrees of freedom in a traditional graph or directed graph. Many models that involve these structures are generative, that is to say that one has a well defined way of prescribing a constructive growth paradigm and studying the complex emergent properties of the resulting states \cite{Wu2015}. However, statistical physics models on the space of simplicial complexes and ASCs with certain structures are becoming more popular \cite{Bianconi2015}. Although work continues to formally understand the topological properties of this space, finding descriptive algorithms with known probability distributions still requires concentrated effort---especially for models that would be computational feasible \cite{Costa2016}.

  \section{Challenges and Solutions in Sampling Abstract Complexes}

  Our goal is to introduce a new sampling algorithm that is both generative and descriptive on the ASC space $C_n$ with a fixed number of nodes $n$ that can then be used for sampling within algorithms that require random walks on this space, such as the oft used Metropolis Algorithms within Markov Chain Monte-Carlo methods employed throughout computational physics. Due to the combinatorial explosion, the cardinality of this space becomes very large very quickly with increasing $n$. Richard Dedekind in 1897 was the first to count the number of such configurations, as labeled ASCs are related to monotone boolean functions \cite{Dedekind1897}. Dedekind numbers, which count the number of ASCs with $m$ elements, are only known for $m\leq 8$; however, asymptotic formulas are also known for large $m$. For the purposes of sampling the unique (nonisomorphic) configurations in the space, we need to remove the labeling that introduces equivalence classes of states under label automorphisms. The inequivalent state cardinalities (and their asymptotic forms) are known only for $m \leq 7$, and grow to be on the order of $5\times10^6$ by $m=7$ \cite{Stephen2014}. We note that these numbers provide an upper bound on $|C_n|$, as they also include nodal removal. Nevertheless, efficiently sampling such a high dimensional space, especially given the equivalence classes, is a challenge. Since there is not yet a general way to know the cardinalities of the isomorphism classes of simplicial complexes on $n$ nodes, we can do little to tune our algorithm to accommodate this directly. Furthermore, designing either a reversible walk or a sampler with known transition probabilities on such a constrained space is an additional challenge that we face.

  In the sections to follow, we introduce two new algorithms for sampling on $C_n$. We design some basic guiding principles that we show analytically yield a non-local uncorrelated fully ergodic sampler that exhibits extremely strong sampling properties. We numerically illustrate its fast and wide sampling capabilities in comparison to a benchmark model. We also design a local ergodic random walk with known transition probabilities that, at the cost of autocorrelation, samples even more efficiently. Lastly, we characterize the autocorrelation of the walk, and numerically test it against our sampler.

\section{Notation and State Visualization}

  As there are a variety of ways to encode the data of a state $C\in C_n$, we take the opportunity to clarify for the reader the representation we will work with.

  \begin{definition}[Digraph Representation G]\label{eq:graphConditions}

    A state $C \in C_n$ is expressed in a leveled digraph representation $G[C]$ if each node $\alpha_d$ in the digraph at level $d$ represents a $(d-1)$-simplex in $C$, with $\alpha$ as a member of the indexing set on level $d$, $\alpha \in [1,|\{\alpha_d\}|]$.
    Defining the set $\{\alpha_1\}$ to the be `roots' of the graph with no incoming edges, the directed adjacency structure is constrained such that the following conditions are satisfied:

  \begin{enumerate}
    \item Directed edges exist only between levels $d \rightarrow (d+1)$
    \item The number of parents of node $\alpha_{d>1}$ must be $d$
    \item The number of roots corresponding to the union of the heads of all dipaths leading to $\alpha_d$ must be $d$
  \end{enumerate}

  \end{definition}

  The last condition guarantees simplicial closure, such that for each simplex, its boundary set are also nodes in the graph state with the proper completeness. There can be at most $\binom{n}{d}$ nodes in a level, corresponding to the ASC that is the d-skeleton of the complete clique complex on $n$ nodes. Similarly, the maximum level is $d=n$.

  This graph representation encodes an ASC uniquely up to $\alpha$ labeling. We denote the \emph{geometric} state as one in which the labeling has been removed. For an example of a labeled state with a canonical ordering, we illustrate in Fig. \ref{fig:ASCGraph} the complete state on 3 roots corresponding to a 2-simplex.

  \begin{figure}[H]
    \centerline{\includegraphics[scale=2]{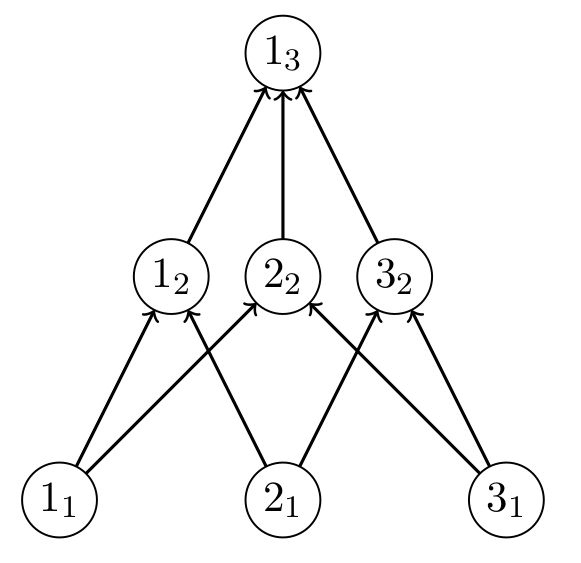}}
    \caption{A Representative Graph State Corresponding to a 2-simplex \label{fig:ASCGraph}}
  \end{figure}

  The convenience of this representation allows us to repackage the boundary closure constraints into the adjacency structure of this digraph, with the directed nature proving useful for easily identifying branching subgraphs affected by said closure.

  \begin{definition}[Boolean Map]
  Let the complete state on $n$ nodes be denoted $C^*_n$. A boolean representation of $C^*_n$ is given by an ALL-TRUE vector with length $\sum_{k=1}^n {n \choose k} = 2^n -1$, where the elements of the vector correspond to a level-canonical ordering of nodes in $G[C]$:  $\alpha_d = 1$ indicating existence of node $\alpha_d$ and $\alpha_d = 0$ indicating non-existence.

    Define $F\,:\,C^*_n \mapsto C$ as a boolean function that assigns $0\vee1$ to each $\alpha_d \in C^*_n$ such that the conditions in Def. \ref{eq:graphConditions} are satisfied.
  \end{definition}

  It is trivial to see that the space of all such functions $\mathbb{F} \ni F$ covers $C_n$. $F$ provides an arbitrary labeled ASC in the boolean representation that can be again visualized through the graph representation $G[C]$ and can be thought of as a mask on $C^*_n$. Isomorphic states are related by boolean functions equivalent up to subset permutations preserving the constraints.

  On $C_3$ for example, the masks $[1 1 1 1 1 0 0]$ and $[1 1 1 1 0 1 0]$ correspond to the same geometric state and can be shown to be equivalent through the allowed subset permutation on the elements corresponding to level $d=2$.

\section{Kahle's Inductive Construction}

  Kahle recently introduced a construction for generating random ASCs \cite{Kahle2014}. We describe some of its properties here, using our above notation for consistency.

  \begin{definition}[Kahle's Model]
  Kahle's multi-parameter model $\Delta(n,p_2,\ldots,p_n)$ builds an ASC inductively, starting at the edge set with $d = 2$. For every $\alpha_d$, include the simplex with probability $p_d$ provided it satisfies the boundary conditions in Def. \ref{eq:graphConditions}.
  \end{definition}

  The full state is built level by level, with constraints on the allowable set of nodes one can include at a given level due to the boundary existence requirements induced by the lower levels.

  Let $|\alpha'_d|$ indicate the number of included simplicies at level $d$ and $|\alpha^*_d|$ indicate the number of possible simplicies given the $(d-1)$ structure:
  \begin{equation*}
    |\alpha'_d|\leq|\alpha^*_d|\leq {n \choose d} \, .
    \end{equation*}

    A labeled state $G[C]$ is generated with probability $P_\Delta$ given by the following:
    \begin{equation*}
      P_\Delta (C) = \prod_{d=2}^{n} {p_d}^{|\alpha'_d|} (1-p_d)^{|\alpha^*_d|-|\alpha'_d|} \, .
    \end{equation*}

  As shown by Zuev \emph{et al.}, Kahle's model is an Exponential Random Simplicial Complex, implying that it generates a maximum entropy ensemble for an expected number of simplicies in the skeletal structures (directly constrained by the probability parameters) \cite{Zuev2015}.

  We note that the probability of achieving a particular state decreases as a binomial power in the number of total nodes in $G[C]$. Even under a nonuniform probability weighting of the levels, it can be easily seen that the combinatorial multiplicities of nodes in each level create a sampling that is highly peaked around states with a given maximum level for large $n$---either one that terminates early at the lower levels leaving no higher structures, one that does the opposite, or one that samples toward the `half-graph' state with $\approx \lceil n/2 \rceil$ levels in the case when we take the probabilities to be coin flips. Precise fine tuning would be needed to allow for sampling across a stretch of widely differing geometries, and the power behavior for finding a particular state will still not be mitigated. Additionally, the isomorphism classes of geometric states will be sampled from with additional probability factors based on their sizes. As the number density of labeled states concentrates toward those that terminate at the central level, we will take the model $\Delta(n,\frac{1}{2},\ldots,\frac{1}{2})\equiv \Delta_\frac{1}{2}$ to benchmark against. Such an algorithm has a probability lower bound at the complete state as follows:

  \begin{equation*}
    \tilde{P}_{\frac{1}{2}} \equiv \min_{C\in C_n} P_{\frac{1}{2}}(C) = P_{\frac{1}{2}}(C^*_n) = (\frac{1}{2}) ^ {\sum_{d=2}^n {n\choose d}} = (\frac{1}{2}) ^ {2^n -n -1} \,.
  \end{equation*}

  We note Kahle's construction was never claimed to be a fast and broad sampler on $C_n$. However, from the class of both descriptive and generative algorithms, and as a producer of a maximum entropy ensemble, it is an incredibly simple and natural inductive construction that we feel would serve as a reasonable baseline to compare against our random sampler on this space with the goal of rounding small probability sets in mind.

\section{The Balanced Algorithm}\label{sec:balancedAlg}

  Our goal is to sample across geometrically inequivalent states with better mixing than the $\Delta_\frac{1}{2}$ model. To this end, we define three key properties that we wish our model to satisfy as heuristics that we intuitively suggest would promote more rapid and broad sampling.

  \begin{enumerate} \label{balancedConditions}
    \item \label{balancedConditions1} Any isolated node such that $|\alpha^*_d| = 1$ should be given a probability of appearance of $p_d = \frac{1}{2}$. At this level in the induction, there are only two possible states that can be selected as the rest of the structure is already fixed. Each state should be given equal probability, as from the vantage of the current step in the algorithm, there is no differentiating property of either state that would induce a bias in the probability. For example, the highest dimension simplex should always have $p_n = \frac{1}{2}$.
    \item \label{balancedConditions2} The power law behavior of binomials in the probabilities should be avoided for individual states, which may also aid the associated issue in over-selecting multiple isomorphic states.
    \item \label{balancedConditions3} The completely disconnected state on $n$ nodes, $C^o_n$, should have the same probability of occurrence as $C^*_n$. This heuristic aims to re-balance the combinatorial effects of the intimate coupling between nodes at different levels due to simplicial closure, since not including any nodes at $d=2$ generates $C^o_n$, while in a construction like $\Delta$, all nodes in $G[C]$ must be independently kept to generate $C^*_n$, regardless of what probabilities are assigned to each level or even each individual simplex.
  \end{enumerate}

  To accomplish this, we first note that we will work inversely from Kahle's inductive constructive model and instead consider an equivalent inductive destructive model. Instead of starting from $C^o_n$, we start from $C^*_n$ and remove nodes starting at $d=2$ and work upwards in level. This is equivalent to sampling on the space $\mathbb{F}$, inductively building the boolean mask starting from the all-ones vector. This is computationally easier, as instead of checking the complicated closure conditions at each node we would like to place, we only have to solve for the complete graph state once (which involves finding all complete subgraphs on $n$ nodes, the NP-complete clique problem), save this state to disk, and reference it at will. To retain the simpliciality, upon removing node $\alpha_d$, one simply removes the unique directed tree associated with $\alpha_d$ as a starting node, which is a linear-time computation. In practice, this amounts to inductively applying a logical AND between the active masking function $F$ and the logical vector $\text{NOT}[\text{IN TREE}]$ for the removed head node.

  \begin{theorem}
    Let $\vec{P_d} = [P_{d_0},P_{d_1},\ldots,P_{d_{{n \choose d}-\hat{d} }}]$ be a probability vector such that $\|\vec{P}_d\|_1 = 1$ with $P_{d_i}$ denoting the probability that $i$ nodes are chosen uniformly at random and removed from level $d$, and $\hat{d}$ indicating the number of nodes already removed from level $d$ due to directed tree pruning from lower level removals.
    \begin{eqnarray}\label{eq:balancedProbs}
      P_{d_{i\neq 0}} &\equiv& P_d = \frac{1}{1+ \sum_{k=d}^n { {n \choose k} - \hat{k} } } \\
      P_{d_0} &=& 1 - ({n \choose d} - \hat{d}) P_d \nonumber
    \end{eqnarray}
    satisfies all properties of conditions \ref{balancedConditions1}, \ref{balancedConditions2} and \ref{balancedConditions3}.
  \end{theorem}

  \begin{proof}

    We first note that $0 < P_d \leq \frac{1}{2} \,\forall\, d$, as the total node set is positive, finite, and the maximum is achieved in condition \ref{balancedConditions1} as proven below. Additionally, $\hat{d}$ is defined such that ${n \choose d} - \hat{d} \in \mathbb{N}$. We need to show that $0 < P_{d_0} < 1$ to conclude that this is a valid probability vector element.

    We only seek to show that $({n \choose d} - \hat{d}) P_d < 1$, as we already know this quantity is strictly positive due to above arguments. It should be clear that $P_d$ is inversely proportional to the total number of nodes left in the state $G[C]$ at step $d$ in the inductive construction ($+1$). The combinatorial prefactor is simply the total number of nodes remaining on level $d$, which must be less than or equal to the total number of nodes in the state. Hence, our claim is justified.

    Lastly, we can safely conclude that $\|\vec{P}_d\|_1 = 1$ by our construction of $P_{d_0} = 1 - ({n \choose d} - \hat{d}) P_d$.

    To show that this distribution satisfies the condition \ref{balancedConditions1}, it can be seen from the definitions that
    \[ |\alpha^*_d| = 1 \Leftrightarrow \hat{k} = \begin{cases}
           {n \choose d} - 1 & k=d  \\
           {n \choose k} & k>d
       \end{cases} \,.
    \]
    Hence,
    \begin{eqnarray*}
      P_d|_{|\alpha^*_d| = 1} &=& \frac{1}{1+{n \choose d} - {n \choose d} + 1 + \sum_{k=d+1}^n {(  {n \choose k} - {n \choose k} } ) } \nonumber \\
          &=& \frac{1}{2} \, ; \\
          P_{d_0}&=& 1- P_d = \frac{1}{2} \,. \nonumber
    \end{eqnarray*}
    Condition \ref{balancedConditions2} is satisfied by algorithmic construction. In choosing groups of $i$ nodes uniformly at random to remove from level $d$, we trade the power-binomial behavior in the probabilities that grow with the number of total nodes in $G[C]$ for a polynomial-binomial behavior that grows with the number of \emph{levels} instead. Additionally, the $\Delta$ model will always pick out a specific labeled $G[C]$ insensitive to the number of isomorphic reachable graphs. In the balanced model, we select from a class of graphs with a certain number of simplicial elements. Although there can also be many such graphs that are not isomorphic but have the same number of elements of given dimensions, we sample the number of elements per level uniformly instead of with product probabilities, giving a key advantage in sets of small probability measure as will be seen exactly in the case of $n=3$ shown in Section \ref{sec:properties}.

    Satisfying condition \ref{balancedConditions3} requires that the removal of all nodes at the edge level have the same probability as removing no nodes at any level:
   \begin{equation*}
     P_2 = \prod_{d=2}^n P_{d_0} \,.
   \end{equation*}
   On the left-hand side,
   \begin{eqnarray}\label{eq:3LHS}
    P_2 &=& \frac{1}{1+ \sum_{k=2}^n { {n \choose k} - \hat{k} } } |_{\hat{k}=0} \nonumber \\
     &=& \frac{1}{1+ \sum_{k=2}^n { {n \choose k} } } \\
     &=& \frac{1}{2^n -n} \nonumber \, .
  \end{eqnarray}
  On the right-hand side,
  \begin{eqnarray}\label{eq:3RHS}
   \prod_{d=2}^n P_{d_0} &=& \prod_{d=2}^n (1 - (\binom{n}{d}-\hat{d})P_d )|_{\hat{d}=0} \nonumber \\
    &=& \prod_{d=2}^n (1 - \frac{\binom{n}{d} } {1+ \sum_{k=d}^n  {n \choose k}} ) \\
    &=& \prod_{d=2}^n \frac{1 + \sum_{k=d}^n \binom{n}{k} - \binom{n}{d} } {1+ \sum_{k=d}^n  {n \choose k}} \nonumber \\
    &=& \prod_{d=2}^n \frac{1 + \sum_{k=d+1}^n \binom{n}{k}} {1+ \sum_{k=d}^n  {n \choose k}} \nonumber \\
    &=& \frac{1} {1+ \sum_{k=2}^n  {n \choose k}} \nonumber \\
    &=& \frac{1}{2^n -n} \nonumber \, .
 \end{eqnarray}
  Comparing Eq. \ref{eq:3LHS} and \ref{eq:3RHS} demonstrates equality.

\end{proof}

  We mention that the existence of such a solution to these constraints is very nontrivial. For example, the balancing condition \ref{balancedConditions3} can be shown to have no solution for the $\Delta_{\frac{1}{2}}$ construction for $n>2$ as equal probability of removal and acceptance would clearly require a solution to an equation of the form
  \begin{eqnarray*}
    x &=& x y  \\
    \text{s.t.}&\quad& 0 < \{x,y\} < 1 \nonumber \,.
  \end{eqnarray*}
  Since $n=2$ doesn't admit more than one probability level (equivalently let $y=1$), the conditions admit the trivial solution $x=\frac{1}{2}$.

  For any constant probability model $\Delta_x$ on $n$ roots enforcing the balancing condition \ref{balancedConditions3} and condition \ref{balancedConditions1} requires the probabilities to be roots of polynomials of the form
  \begin{eqnarray*}
    x^{\frac{n(n-1)}{2}} &=& \frac{1}{2}(1-x)^{2^n-n-2} \\
    \text{s.t.}&\quad& 0 < x < 1 \nonumber \,.
  \end{eqnarray*}
   The computer algebra package Mathematica suggests that this equation does not have any rational solutions for $x$ with $n>2$, indicating that there is likely no natural combinatorial factor that can be attributed to the probability weighting for this model, and relaxing condition \ref{balancedConditions1} does not help.

   For a generic $\Delta(n,p_2,\ldots,p_n)$ model, our constraints require parameters that satisfy the following equation:
   \begin{eqnarray*}
     p_2^{\binom{n}{2}} &=& \frac{1}{2} \prod_{d=2}^{n-1} (1-p_d)^{\binom{n}{d}} \\
     \text{s.t.}&\quad& 0 < p_d < 1 \nonumber \,.
   \end{eqnarray*}

   In the generic case with independent level probabilities, rational solutions only appear to exist if we remove condition \ref{balancedConditions1}; however, this may lead to an large imbalance in the state probabilities for states that are otherwise inductively identical---taking us further from our goal of uniformly sampling the geometric states. It is clear that although possible in theory to balance this algorithm, it requires finding numerical roots at each order and tuning the probabilities to best counteract the power behavior in the sampling, unlike the version we have presented that has closed-form analytic balancing and naturally handles the power structure.

  We conclude this section with the probability of finding a given labeled state using this algorithm. As mentioned, this algorithm samples from classes of complexes with certain numbers of objects per skeletal level. In order to relate these probabilities to a specific geometric state, one must know how these classes decompose into nonisomorphic graphs, as well as the relative sizes of the equivalence classes, introducing an additional combinatorial factor.

  Let the set of all graph isomorphisms between representations $G$ of a geometric state $C$ be denoted $\text{ISO}(G[C])$ such that the cardinality of this set gives the number of equivalent ways of representing $C$ under Def. \ref{eq:graphConditions}.

  At each inductive step, let $i_d$ nodes be removed from level $d$ out of the total number of available nodes.

  The fraction given by the number of labeled ways the selection can be made, weighted by the number of equivalent states at that level, yields the leveled combinatorial factor. Multiplying these factors over the full induction yields the resulting combinatorial factor $\xi(C)$ for achieving a particular geometric state:

  \begin{equation*}
    \xi(C) = |\text{ISO}(G[C])| \prod_{d=2}^n \frac{1} { \binom{\binom{n}{d} - \hat{d}}{i_d}   } \, .
  \end{equation*}
   However, since $|\text{ISO}(G[C])|$ is not known in advance, we can only compute probabilities analytically for labeled states as this breaks the symmetry factor. Thus, the combinatorial factor becomes
 \begin{equation*}
    \xi_L(C) = \prod_{d=2}^n \frac{1} { \binom{\binom{n}{d} - \hat{d}}{i_d}   }  \, .
 \end{equation*}

  It is this quantity that we will use in our comparisons to the $\Delta$ model, as they both consider specific labeled states. In practice, the geometric probabilities are larger, with the labeled probabilities providing a lower bound.

  Let $\{j\}$ be a boolean sequence representing whether any nodes were masked from $C^*_n$, with $j_d \equiv \{j\}_d = 0$ as an indicator that no nodes were removed from level $d$. In terms of our boolean function $F$, the elements correspond to a $\text{NOT}[\text{ALL}[F_d]]$ operation over the level subsets $F_d \subset F$.
  The probability of finding a labeled state is given by the following expression:

  \begin{equation*}
    P(C) = \xi_L(C) \prod_{d=2}^n (P^d)^{\delta^1_{j_d}} (P^d_0)^{\delta^0_{j_d}} \,,
  \end{equation*}
  where $\delta^a_b$ is the Kronecker delta.

\section{Properties of the Balanced Algorithm and Simulation Results} \label{sec:properties}

  This algorithm samples across a weighted space of paths for inductively building a given state, as opposed to building a specific state itself. In the case where each such path yields a unique state up to relabeling, this algorithm will produce the uniform distribution on the space of complexes. Such a condition is only true for $n=\{2,3\}$ where $\xi(C)=1\,\forall \, C\in C_{\{2,3\}}$, and is illustrated in a direct comparison with the $\Delta_{\frac{1}{2}}$ benchmark in Fig. \ref{fig:uniform}. This graph bins the multiplicities for which each geometric state was sampled, subtracted by the mean multiplicity to give residuals, and normalized by the total number of samples. The bins themselves do not match to the same geometric state between the two algorithms, but map to the first encountered representative of a given state. One can clearly see the uniform sampling from the balanced algorithm, although given the number of total samples, both algorithms find all $5$ geometric states. All simulations were performed using MATLAB.

    \begin{figure}[H]
      \centerline{\includegraphics[scale=3.1]{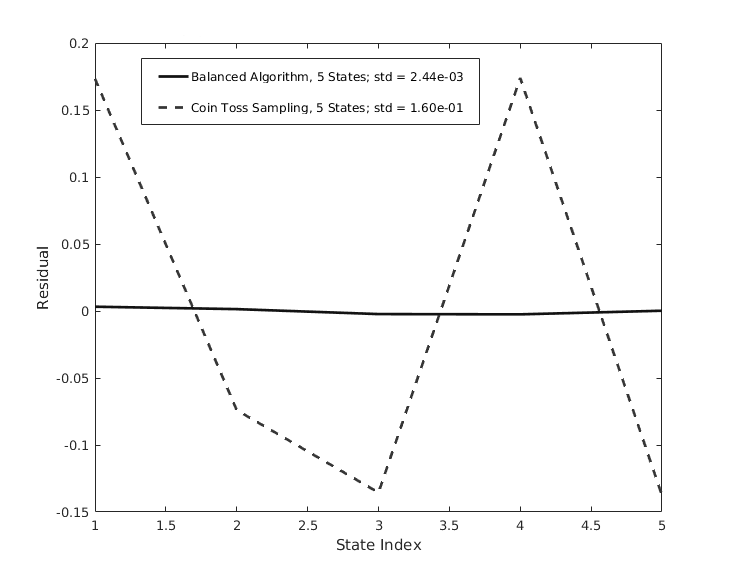}}
      \caption{Multiplicities Residuals of Unique Geometric States on 10000 Samples Drawn From $C_3$, Linearly Interpolated \label{fig:uniform}}
    \end{figure}

  However, for $n=4$ and higher, there exist nonisomorphic graphs with the same number of simplicial elements in each skeleton. This introduces a nonuniform combinatorial factor that is not possible to account for at the time of writing due to the fact that there is no analytic algorithm for predicting the number of such inequivalent graphs and their combinatorial multiplicities. Of course, since we can explicitly compute the probabilities for generating a labeled state, we mention that this sampler can be equipped with a Metropolis filter to re-weight the probabilities to produce a uniform sampling on labeled states.

  We now examine the raw probabilities for sampling a unique labeled state. Directly comparing the minimal probability in the $\Delta_{\frac{1}{2}}$ model with the equivalent complete state in the balanced model indicates that this state has a much greater probability of occurrence:

  \begin{equation*}
    \frac{1}{2^n -n} > (\frac{1}{2}) ^ {2^n -n -1} \,\,\forall \,\, n > 2 \, .
  \end{equation*}

  To indicate whether the new algorithm has balanced the probabilities at large and removed sets of extremely suppressed measure would require looking at the minimal probability bound for this algorithm and comparing it to $C^*_n$ as generated from $\Delta_{\frac{1}{2}}$. Here, we must use the labeled combinatorial factor $\xi_L$ for adequate comparison. Due to the balancing, the probabilities are minimized toward the half-graph state, as this maximizes the binomial coefficients at each level with many combinatorial possibilities equivalent to the removal of certain numbers of nodes. As we would like a lower bound, we set $\hat{k}=0 \, \forall \, k$. Even though we are removing approximately half of the nodes at each level, to maximize the binomial contribution, maintaining the full combinatorial degree of each level will further decrease the probabilities.

  In total, this gives an estimate for a lower bound of the following form:

  With
       \[ E(x) = \begin{cases}
             \frac{x}{2} & \mod(x,2) = 0  \\
             \frac{x+1}{2} & \mod(x,2) = 1
         \end{cases} \,,
      \]
  \begin{equation*}
    \tilde{P} = \min_{C\in C_n} P(C) \approx \prod_{d=2}^{E(n)} \frac{1}{ \binom{\binom{n}{d}}{E(\binom{n}{d})}  \frac{1}{1+\sum_{k=d}^n \binom{n}{k}} } \,.
  \end{equation*}

  Numerical analysis confirms that $\tilde{P}_{\frac{1}{2}} \leq \tilde{P} $ for reasonable values of $n$ before they become numerically unstable due to the combinatorial explosion, as illustrated in Fig. \ref{fig:probAnalysis}. It is immediately apparent that this algorithm has a much stronger probability behavior and actively works against the suppression found in a product model.

  \begin{figure}[H]
    \centerline{\includegraphics[scale=3]{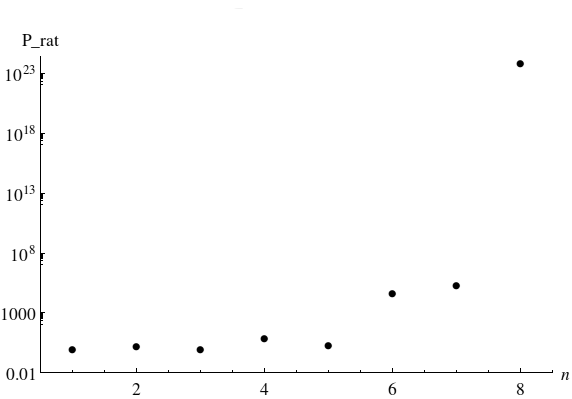}}
    \caption{A Log-Plot of the Ratio $\frac{\tilde{P}}{\tilde{P}_{\frac{1}{2}}}$ as a Function of the Number of Roots $n$ \label{fig:probAnalysis}}
  \end{figure}

  Lastly, we advertised that the combinatorial balancing would allow for a broader access of states. Below we provide some simulation results to illustrate this property. Fig. \ref{fig:numUniqueHit} shows the number of unique geometric states encountered while sampling $C_6$ for a variety of sampling lengths. We can see that the balanced algorithm samples states at a faster rate than the $\Delta_{\frac{1}{2}}$ benchmark test. This is again demonstrated in Fig. \ref{fig:multiplicities}, where 50000 samples were drawn on $C_5$. The balanced algorithm has appeared to converge, while the $\Delta_{\frac{1}{2}}$ benchmark has yet to find all of the inequivalent states.  Naturally, the states with higher probability of being encountered were among the first to be sampled, explaining the correlation between the large initial fluctuations in the two algorithms given the first-representative binning process. However, the multiplicity fluctuations are much smaller for the balanced algorithm, indicating that the goal of heuristically rounding the space of state probabilities has been preliminarily accomplished by this algorithm.

  \begin{figure}[H]
    \centerline{\includegraphics[scale=3.5]{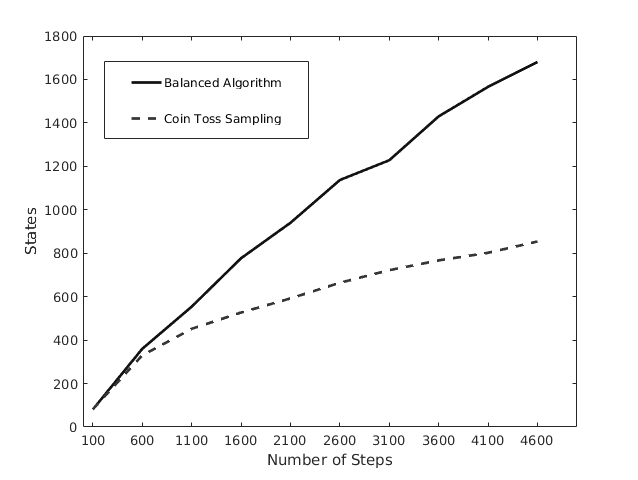}}
    \caption{Number of Unique Geometric States While Sampling $C_6$ as a Function of the Sample Size, Linearly Interpolated \label{fig:numUniqueHit}}
  \end{figure}
  \begin{figure}[H]
    \centerline{\includegraphics[scale=3.5]{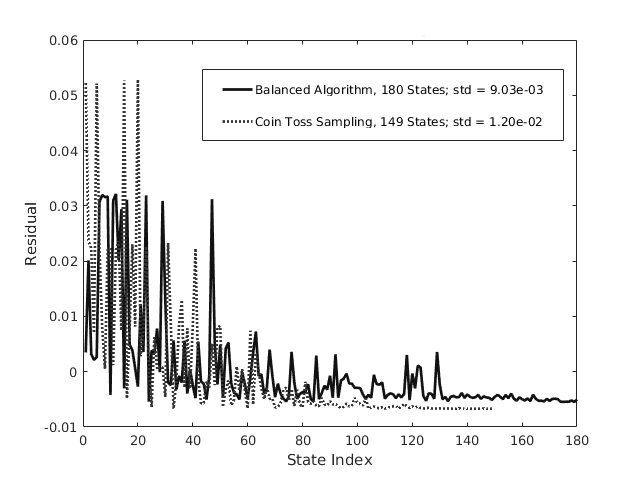}}
    \caption{Multiplicity Residuals of Unique Geometric States on 50000 Samples on $C_5$, Linearly Interpolated \label{fig:multiplicities}}
  \end{figure}

\section{Local Random Walks}

The algorithm introduced in Sec. \ref{sec:balancedAlg} can be naturally used to perform an ergodic walk on the ASC space. We can jump from any state to another without a barrier as there is no dependence on the current state to restrict the space of next available states. This is a desirable feature from the perspective of sampling on the full space, as there are no regions of low conductance in the state space where our `walk' can become trapped. However, when a Metropolis filter is utilized, the fact that this sampler can introduce transitions between arbitrary configurations may be a detriment to the acceptance rate if the filter is not naturally tuned to the intrinsic sampling probabilities of the walk. A local random walk between nearby configurations would be more likely to permit an acceptance with respect to a Metropolis filter, and as a result, may sample the full space faster. We would still like to be able to make this local random walk reversible, however, to ensure the Markov property. However, on the ASC space, the closure constraints make constructing a reversible local random walk very difficult. It is even still difficult to find \emph{any} local random walk where one can compute forward and backward transition probabilities in order to force the walk to be reversible with respect to an additional Metropolis filter. In the following sections, we illustrate one example of a local random walk in the ASC space, and compare its properties to our global sampler.

\section{A Local Random Walk on Abstract Simplicial Complexes}

Our goal is to create a local random walk that has more favorable acceptance ratios for a Metropolis filter that is in some way sensitive to the topological structure of the current ASC. To that end, we define `local' with respect to a new metric on the ASC space that is restricted to a unidirectional walk away from a given state.

\begin{definition}[Unconstrained Nodes]
  Given an ASC in the graph representation $G[C]$, we define an `unconstrained node' \emph{u} as one that can be freely added or removed without requiring or destroying additional containment structure.

  An unconstrained node is `removable' if is has no children and is itself not a root (as the we hold the roots fixed in $C_n$).

  An unconstrained node is `addable' if it is a member of $C^*_n \setminus C$ (the complete graph with the current state excluded) that has all of its parents in $C$.
\end{definition}

We work with unconstrained nodes for two reasons. Foremost, we would like to have a walk that admits a range of local movement as opposed to simply a one step nearest-neighbor walk on individual simplices. If we admit moves that can add or remove an arbitrary number of nodes within the state space, one needs to worry about the closure constraints. These constraints will make it very difficult to generate a walk that has computable probabilities for reversibility, as the number of admissible additions or removals would be dynamic with each sub-step within the same transition move, and there can be multiple paths with different probabilities that could lead to the same state. We want to restrict down this capability, but still admit larger jumps through the state space. Hence, we work with the space of unconstrained nodes as pure additions or removals within this space will prevent such issues from arising and admit a walk with computable probabilities. The restriction that nodes are only added or removed in a single step additionally guarantees that we do not have any closed loops within our multi-step walk for a given transition.

Our notion of local distance is therefore the number of added or removed nodes in a given transition step, actioned by a binary flip on the boolean function representing $C$.

The algorithm mimics an exponential ball walk with respect to this distance measure. First, we compute the total number of nodes one could maximally flip on the state space. From this set, we establish a normalized probability function based on an exponential decrease in probability for larger numbers of binary flips. We decide to either add or remove nodes in a given transition. Once this choice is made, we compute all addable/removable nodes $U$ for the current configuration. We select a distance $\delta$ to move based on the fixed probability measure. If that distance takes the state outside of the state space or beyond the number of admissible adds/removes, then the algorithm resets until an admissible move is found---this is our rejection sampling step similar to a ball walk on the edge of the state space. Once a good distance is accepted, a uniformly random selection of those unconstrained nodes $u\in U$ have their entries flipped in the boolean representation. The forward and backward probabilities are symmetric with respect to the exponential distance weighting, as this is not dependent on the state itself. Therefore, a Metropolis filter would only need to account for the uniform selection step, producing a combinatorial factor of

\begin{equation}\label{eq:walkProbs}
  \frac{P_B}{P_F} = \frac{{|U_B| \choose \delta}}{{|U_F| \choose \delta}} \, .
\end{equation}

\section{Computational Results}

Since our sampler now has local correlations, it becomes necessary to characterize more carefully the efficiency of the random walk and breadth of sampling.
We present two extreme situations for the initial start of the walk: beginning at a corner of the state space, $C^*_n$, and beginning at a `central' state consisting of roughly half of the available simplices being activated. We examine both the multiplicity residuals as before, as well as the autocorrelation length.

To characterize the autocorrelation length, we use an initial convex sequence method that involves the greatest common minorant \cite{Geyer2011}. First, we implement a Metropolis filter utilizing Eq. \ref{eq:walkProbs} such that our samples can be expected to be i.i.d. To measure autocorrelation, we compute the signed displacement of a transition between two states $C$ and $C'$ as the difference in the sums of their boolean representations:

\begin{equation*}
  \delta = \|C\|_1 - \|C'\|_1 \,,
  \end{equation*}
where $|\delta|$ still corresponds to the number of binary flips between the two, as discussed in the algorithm. We then look at the cumulative sum of the time sequence of $\delta_i$ values for each step $i$ in the walk. This gives some sense of a $1$-dimensional projection of the random walk through the ASC space, making it a natural random variable to compute autocorrelations with.

Let an estimator for the sample mean on $s$ samples be denoted

\begin{equation*}
  \hat{\mu}_s = \frac{1}{s} \sum^n_{i = 1}{\delta_i} \,.
  \end{equation*}

A natural estimator for the auto-covariance function at lag $k$ is given by the following:

\begin{equation*}
  \hat{\gamma}_s = \frac{1}{s} \sum^{n-k}_{i = 1}{(\delta_i - \hat{\mu}_s)(\delta_{i+k} - \hat{\mu}_s)} \,.
  \end{equation*}

The greatest convex minorant at lag $k$,

\begin{equation*}
  \Gamma_k = \gamma_{2k} + \gamma_{2k+1},
\end{equation*}

is a strictly positive, decreasing, and convex function for a reversible Markov chain. Therefore, examining our estimator $\hat{\Gamma}_k$ for the point at which it becomes nonpositive indicates the lag where we encounter autocorrelation. Due to the dependence on twice the lag, our autocorrelation is related to $2k'$ when $\hat{\Gamma}_{k'} \leq 0$.

We can see in Fig. \ref{fig:randomWalkCenterAC} that the local random walk started from a central state, upon re-weighting with the Metropolis filter, produces autocorrelation out to about $16$ steps, and the walk has a natural rejection rate on the order of $50\%$.

\begin{figure}[H]
  \centerline{\includegraphics[scale=.5]{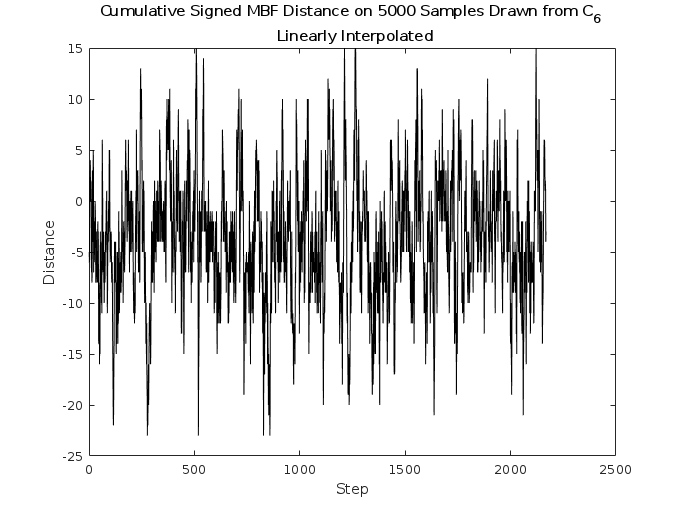}}
  \caption{Time History Observable Example Used to Compute Autocorrelation}
  \end{figure}
  \begin{figure}[H]
  \centerline{\includegraphics[scale=.5]{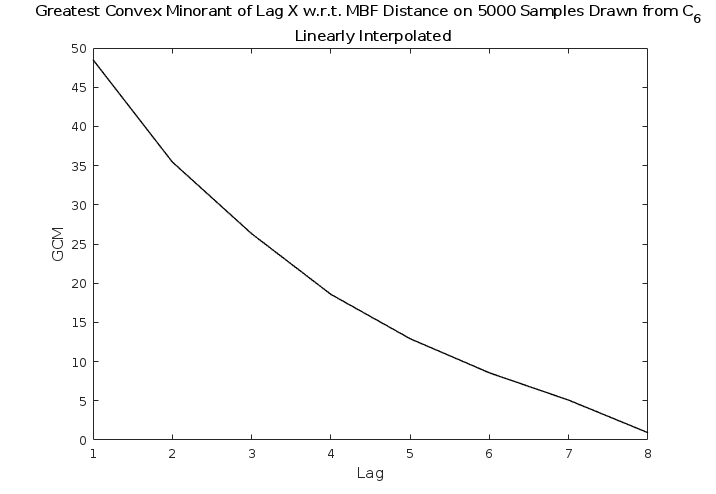}}
  \caption{Autocorrelation Statistics for a Random Walk on $C_6$ with $5000$ Steps, Starting from a Central State}
  \label{fig:randomWalkCenterAC}
\end{figure}

Fig. \ref{fig:randomWalkCornerAC}, produced starting from a corner state, tells not much of a different story. This indicates that the edges in the state space are not incredibly narrow, and that this random walk is good at working its way out of those corners. We see less than double the autocorrelation, which is not unexpected due to the time spent in the region of small state density. A burn-in process would reduce this down to toward the autocorrelation lengths found in the central case.

\begin{figure}[H]
  \centerline{\includegraphics[scale=.5]{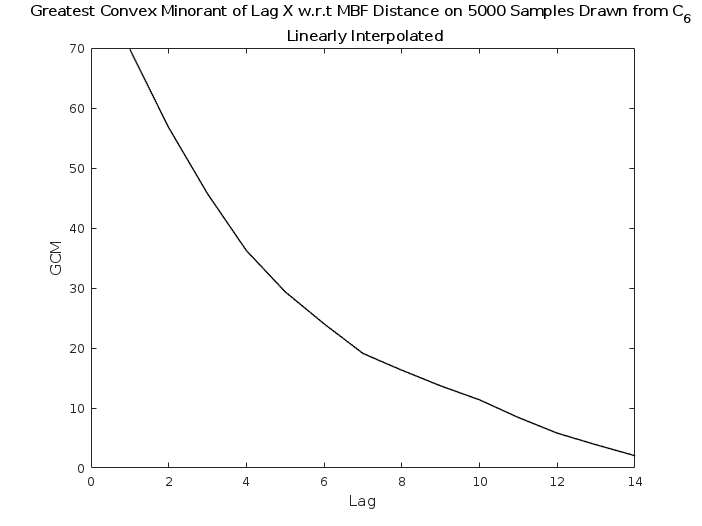}}
  \caption{Autocorrelation Statistics for a Random Walk on $C_6$ with $5000$ Steps, Starting from a Corner State}
  \label{fig:randomWalkCornerAC}
\end{figure}

We lastly compare the efficiency of all of these algorithms for sampling geometrically unique states. As seen in Figs. \ref{fig:StatesVsSteps} and \ref{fig:multiplicitiesAll}, the local random walk performs remarkably better, sampling more states with less accepted steps. This lends credence to the notion that the best sampler on this space would likely be a linear combination of the two Markov chains. Since such a construction still retains its theoretical properties, we can achieve the best of both algorithms by choosing to perform a local walk with some large probability to reap the rewards of the rapid sampling, while occasionally using the balanced sampler to avoid regions of narrow conductance bands and to promote ergodicity and large nonlocal transitions.

\begin{figure}[H]
  \centerline{\includegraphics[scale=.5]{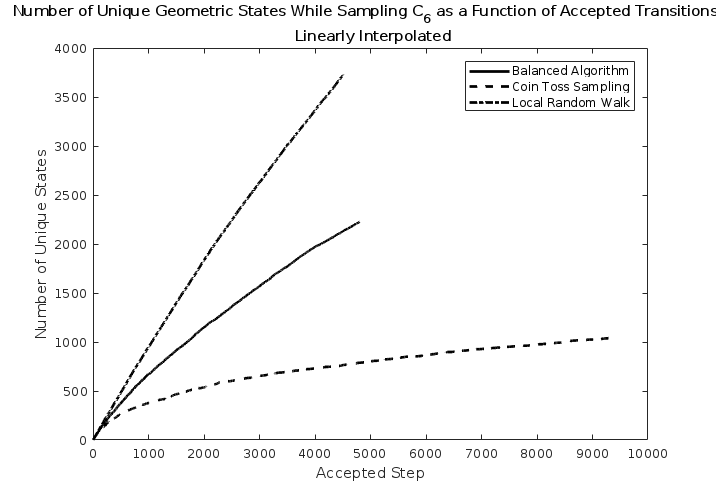}}
  \caption{A Comparison of the Unique Geometric States Sampled on $C_6$ as a Function of Accepted Transitions for All Three Samplers}
  \label{fig:StatesVsSteps}
\end{figure}
\begin{figure}[H]
  \centerline{\includegraphics[scale=.5]{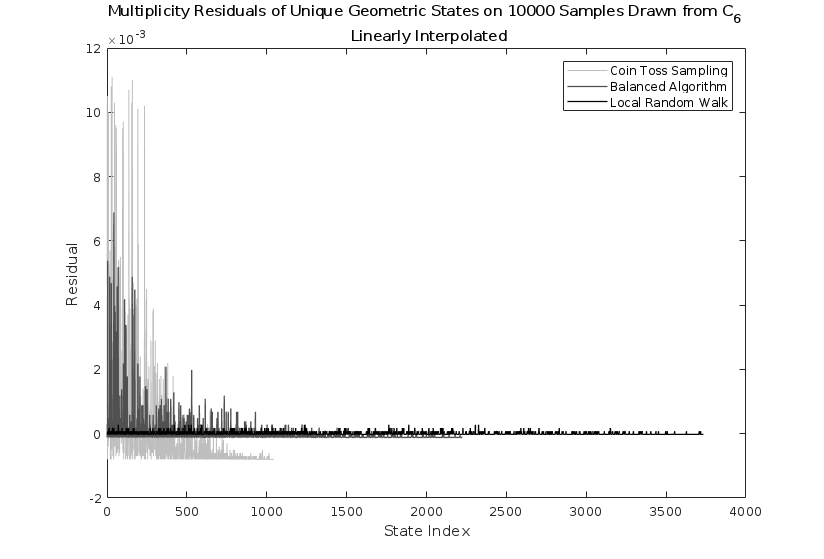}}
  \caption{A Comparison of the Multiplicity Residuals of Unique Geometric States Sampled on $C_6$ for All Three Samplers, Linearly Interpolated \label{fig:multiplicitiesAll}}
\end{figure}

\section{Discussion on the New Samplers}
As is the case with a wide variety of combinatorial spaces, it is often very difficult to develop a sampling procedure with transition probabilities that can \emph{a priori} sample such that the uniform distribution is the stationary distribution without the use of a Metropolis filter. In the case of abstract simplicial complexes, the unknown isomorphism classes of configurations make this problem seemingly intractable. We have introduced an algorithm that uses three simple principles to attempt to re-balance the sampling such that the algorithm more readily samples inequivalent configurations with a wide breadth across the space. Our analytical results show that this algorithm has a worst case lower-bound on state probabilities that is larger than the equivalent sampling through a uniformly weighted Kahle process, which we used as an unoptimized benchmark. Our simulations confirm that a direct comparison between the two algorithms favors the balanced algorithm when attempting to sample across the geometric space of states.

We have also discussed a local random walk that can be made reversible. The advantage of this walk is to increase the acceptance rates for a Metropolis filter when sampling nearby states as opposed to large jumps in the state space, and we have illustrated through simulation its efficiency in also sampling from a wide range of states in the ASC space. However, in some applications with Metropolis filters, this walk may be sensitive to trapping regions, as it is not able to explore any possible configuration in a single transition step. Thus, a combination of our local walk and the balanced sampler can be used to promote ergodicity and rapid sampling.

Future work toward finding a better generative algorithm for sampling across equivalence classes of large random abstract simplicial complexes while maintaining analytical control is necessary in order to begin to probe the very large space of states. With a variety of applications on the horizon, we anticipate this problem being approached from a broad range of perspectives, and we hope to have provided some insight through some practical, simple algorithms that accomplish the first steps toward this task.

\begin{acks}

  This work was supported in part by the University of Washington. We would like to thank Hariharan Narayanan and Stephen Sharpe for their help in revising this manuscript.

\end{acks}

\bibliographystyle{apt}
\bibliography{MyLibrary.bib}

\begin{thebibliography}{1}

\bibitem{Bianconi2015}
{\sc Bianconi, G., Rahmede, C. and Wu, Z.} (2015).
\newblock Complex quantum network geometries: Evolution and phase transitions.
\newblock {\em Phys. Rev. E\/} {\bf 92,} 022815.

\bibitem{Costa2016}
{\sc Costa, A. and Farber, M.} (2016).
\newblock Random simplicial complexes.
\newblock In {\em Configuration Spaces: Geometry, Topology and Representation
  Theory}. ed. F.~Callegaro, F.~Cohen, C.~De~Concini, E.~M. Feichtner,
  G.~Gaiffi, and M.~Salvetti.
\newblock Springer International Publishing, Cham pp.~129--153.

\bibitem{Dedekind1897}
{\sc Dedekind, R.} (1897).
\newblock Uber zerlegungen von zahlen durch ihre grossten gemeinsamen teiler.
\newblock {\em Festschrift der Technischen Hochschule zu Braunschweig bei
  Gelegenheit der 69. Versammlung Deutscher Naturforscher und Arzte\/} 1--40.

\bibitem{Geyer2011}
{\sc Geyer, C.} (2011).
\newblock Introduction to mcmc.
\newblock In {\em Handbook of Markov Chain Monte Carlo}. ed. S.~Brooks,
  A.~Gelman, G.~Jones, and X.-L. Meng.
\newblock Chapman and Hall ch.~Introduction to MCMC.

\bibitem{Kahle2014}
{\sc Kahle, M.} (2014).
\newblock Topology of random simplicial complexes: A survey.
\newblock {\em AMS Contemp. Math.\/} {\bf 620,} 201--22.

\bibitem{Maletic2009}
{\sc Maletic, S. and Rajkovic, M.} (2009).
\newblock {\em Complex Networks (Studies in Computational Intelligence Series)}
  vol.~207.
\newblock Springer-Verlag, Berlin.
\newblock ch.~Simplicial Complex of Opinions on Scale-Free Networks,
  pp.~127--34.

\bibitem{Stephen2014}
{\sc Stephen, T. and Yusun, T.} (2014).
\newblock Counting inequivalent monotone boolean functions.
\newblock {\em Discrete Applied Mathematics\/} {\bf 167,} 15 -- 24.

\bibitem{Wu2015}
{\sc Wu, Z., Menichetti, G., Rahmede, C. and Bianconi, G.} (2015).
\newblock Emergent complex network geometry.
\newblock {\em Nature Scientific Reports\/} {\bf 5,} 10073.

\bibitem{Zuev2015}
{\sc Zuev, K., Eisenberg, O. and Krioukov, D.} (2015).
\newblock Exponential random simplicial complexes.
\newblock {\em Journal of Physics A: Mathematical and Theoretical\/} {\bf 48,}
  465002.

\end{thebibliography}
%
%


\end{document}